\documentclass[a4paper,dvipdfmx]{article}
\usepackage{url,amsmath,amssymb}
\usepackage{times}
\usepackage[override]{cmtt}

\usepackage{tikz-qtree}
\tikzset{
  solid node/.style={circle, draw, inner sep=0.5, fill=black},
}
 
\usepackage{microtype}

\title{On repetitive right application of $B$-terms}

\author{Mirai Ikebuchi\\
        \url{ikebuchi.mirai@a.nagoya-u.jp}\\
        Nagoya University \\
        Aichi, Japan
        \and
        Keisuke Nakano\\
        \url{k.nakano@acm.org}\\
        University of Electro-Communications \\
        Tokyo, Japan
        }

\newtheorem{theorem}{Theorem}[section]
\newtheorem{lemma}[theorem]{Lemma}

\newtheorem{definition}[theorem]{Definition}
\newtheorem{example}[theorem]{Example}

\newcommand\qed{$\Box$}

\newcommand\proofname{\textbf{Proof}}
\newenvironment{proof}[1][\proofname]{\hspace{0pt}\\[-\baselineskip]
\textbf{#1.}~ }{\hfill\qed}

\newcommand\ie{i.e.}

\newcommand\labeqn[1]{\label{eqn:#1}}
\newcommand\refeqn[1]{(\ref{eqn:#1})}
\newcommand\labfig[1]{\label{fig:#1}}
\newcommand\reffig[1]{\textrm{Fig.\,\ref{fig:#1}}}
\newcommand\Reffig[1]{\textrm{Figure~\ref{fig:#1}}}

\newcommand\labthr[1]{\label{thr:#1}}
\newcommand\refthr[1]{\textrm{Theorem~\ref{thr:#1}}}

\newcommand\lablem[1]{\label{lem:#1}}
\newcommand\reflem[1]{\textrm{Lemma~\ref{lem:#1}}}

\newcommand\labcnj[1]{\label{cnj:#1}}
\newcommand\refcnj[1]{\textrm{Conjecture~\ref{cnj:#1}}}
\newcommand\labsec[1]{\label{sec:#1}}
\newcommand\refsec[1]{\textrm{Section~\ref{sec:#1}}}

\newcommand\paren[1]{\left(#1\right)}
\newcommand\sapp[2]{#1_{\paren{#2}}}

\newcommand\CL[1]{\mathbf{CL}(#1)}

\newcommand\nodes{\mathcal{L}}

\newcommand\x{\star}
\def\<#1>{\langle #1\rangle}
\newcommand\poly[1]{%
\langle\dots\langle\langle\x,\underbrace{\x\rangle,\x\rangle,\dots,\x\rangle}_{#1}%
}
\newcommand\size[1]{\left|\!\left|#1\right|\!\right|}
\newcommand\concat{\mathbin{{+}\mspace{-8mu}{+}}}

\newtheorem{conjecture}[theorem]{Conjecture}

\newcommand\Proof{\textit{Proof}}

\begin{document}

\maketitle

\begin{abstract}
$B$-terms are built from the $B$ combinator alone
defined by $B\equiv\lambda f.\lambda g.\lambda x. f~(g~x)$, which is
well-known as a function composition operator.
This paper investigates an interesting property of $B$-terms,
that is, whether repetitive right applications of a $B$-term circulates or not.
We discuss conditions for $B$-terms to and not to have the property
through
a sound and complete equational axiomatization.
Specifically, we give examples of $B$-terms which have the property and show that there are infinitely many $B$-terms which does not have the property. Also, we introduce a canonical representation of $B$-terms that is useful to detect cycles, or equivalently, to prove the property, with an efficient algorithm.
\end{abstract}

\section{Introduction}
\labsec{intro}
The `bluebird' combinator $B=\lambda f.\lambda g.\lambda x. f~(g~x)$
is well-known~\cite{Smullyan12book}
as a bracketing combinator or composition operator,
which has a principal type 
$(\alpha\to\beta)\to(\gamma\to\alpha)\to\gamma\to\beta$.
A function $B~f~g$ (also written as $f\circ g$)
synthesized from two functions $f$ and $g$
takes a single argument to apply $g$ and
returns the result of an application $f$ to the output of $g$.

In the general case where $g$ takes $n$ arguments to pass the output to $f$,
the synthesized function is defined by
$\lambda x_1.\cdots\lambda x_n. f~(g~x_1~\dots~x_n)$.
Interestingly, the function can be expressed by $B^n~ f~ g$
where $e^n$ is an $n$-fold composition of function $e$ such that
$e^0 = \lambda x. x$ and $e^{n+1}=B~e^{n}~e$ for $n\geq0$.
We call {\em $n$-argument composition\/} 
for the generalized composition represented by $B^n$.

Now we consider the $2$-argument composition expressed by
$B^2=\lambda f.\lambda g.\lambda x.\lambda y.~f~(g~x~y)$.
From the definition, we have $B^2 = B\circ B = B~B~B$
where function application
is considered left-associative, that is, $f~a~b=(f~a)~b$.
Thus $B^2$ is defined by an expression
in which all applications nest to the left, never to the right.
We call such an expression {\em flat\/}~\cite{Okasaki03jfp}.
In particular we write $\sapp{X}{k}$
for a flat expression involving only a combinator $X$, which is defined by
$\sapp{X}{1} = X$ and
$\sapp{X}{k+1} = \sapp{X}{k}~X (k\geq 1)$.
Okasaki~\cite{Okasaki03jfp} shows that there is no universal combinator $X$ that can represent
any combinator by $\sapp{X}{k}$ with some $k$.
Using this notation, we can write $B^2 = \sapp{B}{3}$.

Consider the $n$-argument composition expressed by $B^n$.
Surprisingly, we have $B^3 =\allowbreak B~B~B~B~B~B~B~B = \sapp{B}{8}$.
It is easy to check it by repeating $\beta$-reduction
for $\sapp{B}{8}~f~g~x~y~z = f~(g~x~y~z)$.
%
For $n>3$, however, the $n$-argument composition cannot be expressed by flat $B$-terms.
There is no integer $k$ such that $B^n = \sapp{B}{k}$
with respect to $\beta\eta$-equality.
It can be proved by {\em $\rho$-property\/} of combinator $B$,
that is introduced in this paper.
We say that a combinator $X$ has $\rho$-property 
if there exists two distinct integer $i$ and $j$ such that
$\sapp{X}{i} = \sapp{X}{j}$.
If such a pair $i, j$ is found, we have $\sapp{X}{i+k}=\sapp{X}{j+k}$
for any $k\geq 0$
~({\`a}~la \emph{finite monogenic semigroup}~\cite{Ljapin68book}).
In the case of $B$, we can check
$\sapp{B}{6}=\sapp{B}{10}=\lambda x.\lambda y.\lambda z.\lambda w.\lambda v.~x~(y~z)~(w~v)$
hence $\sapp{B}{i}=\sapp{B}{i+4}$ for $i\geq6$.
%
\reffig{rhoB} shows a computation graph of $\sapp{B}{k}$.
The $\rho$-property is named 
after the shape of the graph.
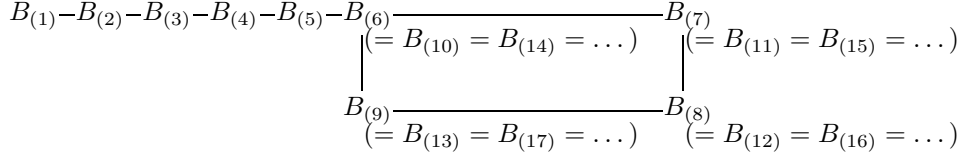
\begin{figure*}[t]\centering
\newcount\picW \newcount\picH \picW=350 \picH=60
\begin{picture}(\picW,\picH)(0,0)
\divide\picW 10 \advance\picW 1 \divide\picH 10 \advance\picH 1 
\def\Line(#1,#2)(#3,#4){
\newcount\middleX \middleX=#1 \advance\middleX #3 \divide\middleX 2
\newcount\middleY \middleY=#2 \advance\middleY #4 \divide\middleY 2
\qbezier(#1,#2)(\middleX,\middleY)(#3,#4)
}
\def\startX{0} \def\upperY{50}
\def\marginB{20} \def\spanBs{5} \def\spanBh{100} \def\spanBv{20}
\newcount\tmpX \newcount\tmpY
\def\putB[#1]{\put(\tmpX,\tmpY){\makebox(20,0){$\sapp{B}{#1}$}} \advance\tmpX 20}
\def\putBs[#1]{\newcount\tmpXs \tmpXs=\tmpX \newcount\tmpYs \tmpYs=\tmpY
\advance\tmpXs-10 \advance\tmpYs-10
\put(\tmpXs,\tmpYs){\makebox(100,0){$(#1)$}}}
\def\lineB(#1){
\newcount\nextX \nextX=\tmpX \advance\nextX #1
\Line(\tmpX,\tmpY)(\nextX,\tmpY) \tmpX=\nextX
}
\def\lineBs{\lineB(\spanBs)} \def\lineBh{\lineB(\spanBh)}
\def\lineBv{
\newcount\nextY \nextY=\tmpY \advance\nextY -\spanBv
\Line(\tmpX,\tmpY)(\tmpX,\nextY) \tmpY=\nextY
}
\def\turnR{\advance\tmpX-12 \advance\tmpY-8}
\def\turnL{\advance\tmpX-8 \advance\tmpY-8}
\tmpX=\startX \tmpY=\upperY
\putB[1] \lineBs \putB[2] \lineBs \putB[3] \lineBs \putB[4] \lineBs \putB[5] \lineBs
\putB[6] \putBs[=\sapp{B}{10}=\sapp{B}{14}=\dots]
\newcount\sixX \sixX=\tmpX \newcount\sixY \sixY=\tmpY \lineBh
\putB[7] \putBs[=\sapp{B}{11}=\sapp{B}{15}=\dots] \turnR \lineBv
\tmpX=\sixX \tmpY=\sixY \turnR \lineBv
\turnL \putB[9] \putBs[=\sapp{B}{13}=\sapp{B}{17}=\dots]
\lineBh \putB[8] \putBs[=\sapp{B}{12}=\sapp{B}{16}=\dots]
\end{picture}
\caption{$\rho$-property of the $B$ combinator}
\labfig{rhoB}
\end{figure*}
The $\rho$-property implies that
the set $\{~ \sapp{B}{k} \mid k\geq1 ~\}$ is finite.
Since none of the terms in the set is equal to $B^n$ with $n>3$
up to the $\beta\eta$-equivalence of the corresponding $\lambda$-terms,
we conclude that there is no integer $k$ such that $B^n =\sapp{B}{k}$.

This paper discusses the $\rho$-property of combinatory terms,
particularly in $\CL{B}$, called {\em $B$-terms\/},
that are built from $B$ alone.
Interestingly, 
$B~ B$ enjoys the $\rho$ property 
with $\sapp{(B~ B)}{52}=\sapp{(B~ B)}{32}$
and so does $B~ (B~ B)$ with
$\sapp{(B~ (B~ B))}{294}=\sapp{(B~ (B~ B))}{258}$
as reported~\cite{Nakano08trs}.
Several combinators other than $B$-terms
can be found enjoy the $\rho$-property,
for example,
$K=\lambda x.\lambda y. x$ and 
$C=\lambda x.\lambda y. \lambda z.~ x~ z~ y$ 
because of $\sapp{K}{3}=\sapp{K}{1}$ and 
$\sapp{C}{4}=\sapp{C}{3}$.
They are not so interesting
in the sense that the cycle starts immediately and its size is very small,
comparing with $B$-terms like $B~ B$ and $B~ (B~ B)$.
As we will see later,
$B~ (B~ (B~ (B~ (B~ (B~ B))))) (\equiv B^6~ B)$ has the $\rho$-property
with the cycle of the size more than $3\times 10^{11}$
which starts after more than $2\times 10^{12}$ repetitive right application.
This is why the $\rho$-property of $B$-terms is intensively discussed
in the present paper.
%

The contributions of the paper are two-fold.
One is to give a characterization of $\CL{B}$ and 
another is to provide a sufficient condition for 
the $\rho$-property and anti-$\rho$-property of $B$-terms. 
In the former, a canonical representation of $B$-terms is introduced
and sound and complete equational axiomatization for $\CL{B}$ is established.
In the latter, 
the $\rho$-property of $B^n B$ with $n\leq 6$ is shown with an efficient algorithm
and 
the anti-$\rho$-property for $B$-terms in particular forms
is proved.

\section{$\rho$-property of terms}
\labsec{prelim}
The $\rho$-property of combinator $X$ is that
$\sapp{X}{i}=\sapp{X}{j}$ holds for some $i>j\geq1$.
We adopt $\beta\eta$-equivalence of corresponding $\lambda$-terms
for the equivalence of combinatory terms in this paper.
We could use other equivalence, for example,
induced by the axioms of the combinatory logic.
The choice of equivalence does not have an essential influence,
e.g., $\sapp{B}{9}$ and $\sapp{B}{13}$ are equal even up to the combinatory axiom of $B$,
while $\sapp{B}{6}=\sapp{B}{10}$ holds for $\beta\eta$-equivalence.
Furthermore, for simplicity,
we only deal with the case where $\sapp{X}{n}$
is normalizable for all $n$.
If $\sapp{X}{n}$ is not normalizable, 
it is much difficult to check equivalence with the other terms.

Let us write $\rho(X)=(i,j)$ 
when a combinator $X$ has the $\rho$-property  
due to $\sapp{X}{i}=\sapp{X}{i+j}$
with minimum positive integers $i$ and $j$.
For example, we can write $\rho(B)=(6,4)$,
$\rho(C)=(3,1)$, $\rho(K)=(1,2)$ and $\rho(I)=(1,1)$.
Besides them,
several combinators introduced in the Smullyan's book~\cite{Smullyan12book}
have the $\rho$-property:
\begin{align*}
\rho(D)&=(32,20) 
&&\text{where $D=\lambda x.\lambda y.\lambda z.\lambda w. x~ y~ (z~ w)$}\\
\rho(F)&=(3,1) &&\text{where $F=\lambda x.\lambda y.\lambda z. z~ y~ x$}\\
\rho(R)&=(3,1) &&\text{where $R=\lambda x.\lambda y.\lambda z. y~ z~ x$}\\
\rho(T)&=(2,1) &&\text{where $T=\lambda x.\lambda y. y~ x$}\\
\rho(V)&=(3,1) &&\text{where $V=\lambda x.\lambda y.\lambda z. z~ x~ y$}\text.
\end{align*}
Except the $B$ and $D~(=B~ B)$ combinators,
the property is `trivial' in the sense that
loop starts early and the size of cycle is very small.

On the other hand,
the combinators $S=\lambda x.\lambda y.\lambda z. x~z~(y~z)$
and $O=\lambda x.\lambda y. y~(x~ y)$ in the book
do not have the $\rho$-property for reason~(A),
which is illustrated by
\begin{align*}
\sapp{S}{2n+1} &=
\lambda x.\lambda y. \underbrace{x~ y~ (x~ y~ (\dots (x~ y}_n~ 
(\lambda z. x~ z~ (y~ z)))\dots))
\text,
\\
\sapp{O}{n+1} &=
\lambda x. \underbrace{x~ (x~ (\dots (x~ }_n~ 
(\lambda y. y~ (x~ y))
\text.
\end{align*}
%

\begin{figure*}
\begin{align*}
\rho(B^0 B)&=(6,4)  
&\rho(B^4 B)&=(191206,431453)\\ 
\rho(B^1 B)&=(32,20) 
&\rho(B^5 B)&=(766241307,234444571)\\ 
\rho(B^2 B)&=(258,36) 
&\rho(B^6 B)&=(2641033883877,339020201163)\\ 
\rho(B^3 B)&=(4240,5796) 
\end{align*}
\caption{$\rho$-property of $B$-terms in a particular form}
\labfig{rhobsequence}
\end{figure*}
The definition of the $\rho$-property is naturally extended
from single combinators to terms obtained by combining several combinators.
We found by computation that
several $B$-terms, built from the $B$ combinator alone,
have a nontrivial $\rho$-property
as shown in~\reffig{rhobsequence}.
The detail will be shown in~\refsec{results}.

\section{Characterization of $B$-terms}
The set of all $B$-terms, $\CL{B}$, is closed under application,
that is, the repetitive right application of a $B$-term
always generates a sequence of $B$-terms.
Hence, the $\rho$-property can be decided
by checking `equivalence' among generated $B$-terms,
where the equivalence should be checked
through $\beta\eta$-equivalence of their corresponding $\lambda$-terms
in accordance with the definition of the $\rho$-property.
It would be useful if we have a simple decision procedure
for deciding equivalence over $B$-terms.

In this section,
we give a characterization of the $B$-terms
to make it easy to decide their equivalence.
We introduce a method for deciding equivalence of $B$-terms
without calculating the corresponding $\lambda$-terms.
To this end, 
we first investigate equivalence over $B$-terms
with examples
and then present an equation system as characterization of $B$-terms
so as to decide equivalence between two $B$-terms.
Based on the equation system, 
we introduce a canonical representation of $B$-terms.
The representation makes it easy to observe 
the growth caused by repetitive right application of $B$-terms,
which will be later shown for proving 
the anti-$\rho$-property of $B^{2}$.
We believe that this representation will be helpful to prove 
the $\rho$-property or the anti-$\rho$-property 
for the other $B$-terms.

\subsection{Equivalence over $B$-terms}
Two $B$-terms are said \emph{equivalent} if their corresponding $\lambda$-terms
are $\beta\eta$-equivalent.
For instance, $B~ B~ (B~ B)$ and $B~ (B~ B)~ B~ B$ are equivalent.
This can be easily shown
by the definition $B~x~y~z = x~ (y~ z)$.
For another (non-trivial) instance,
$B~ B~ (B~ B)$ and
$B~ (B~ (B~ B))~ B$ are 
equivalent.
This is illustrated
by the fact that
they are equivalent to 
$\lambda x. \lambda y. \lambda z. \lambda w.\lambda v. x~(y~z)~(w~v)$.
where $B$ is replaced with $\lambda x.\lambda y.\lambda z.~x~(y~z)$
or the other way around
at the $=_\beta$ equation.
Similarly,
we cannot show equivalence between two $B$-terms,
$B~(B~B)~(B~B)$ and $B~(B~B~B)$,
without long calculation.
This kind of equality makes it hard
to investigate the $\rho$-property of $B$-terms.
To solve the annoying issue,
we will later introduce a canonical representation of $B$-terms.

\subsection{Equational axiomatization for $B$-terms}
Equality between two $B$-terms can be effectively decided
by an equation system.
\begin{figure}
\begin{align}
B~ x~ y~ z &= x~ (y~ z)
\tag{B1}
\\
B~ (B~ x~ y) &= B~ (B~ x)~ (B~ y)
\tag{B2}
\\
B~ B~ (B~ x) &= B~ (B~ (B~ x))~ B
\tag{B3}
\end{align}
\caption{Equational axiomatization for $B$-terms}
\labfig{eqB}
\end{figure}
\Reffig{eqB} shows a sound and complete equation system
as described in the following theorem.

\begin{theorem}
\labthr{eqB}
Two $B$-terms are $\beta\eta$-equivalent 
if and only if
their equality is derived by equations~(B1), (B2), and (B3).
\end{theorem}

The proof of the if-part is given here,
which corresponds to the soundness of the equation system.
We will later prove the only-if-part
with the uniqueness of the canonical representation of $B$-terms.
\\
\begin{proof}[\proofname~of if-part of \refthr{eqB}]
Equation~(B1) is 
immediate from the definition of $B$.
Equation~(B2) and~(B3) are shown by
\begin{align*}
B~ (B~ e_1~ e_2) &= 
\lambda x.\lambda y.~ B~ e_1~ e_2~ (x~ y)
                                    & B~ B~ (B~ e_1) 
                                    &= \lambda x.~ B~ (B~ e_1~ x)
\\&=
\lambda x.\lambda y.~ e_1~ (e_2~ (x~ y))
                                    &&= \lambda x.\lambda y.\lambda z.~ B~ e_1~ x (y~ z)
\\&=
\lambda x.\lambda y.~ e_1~ (B~ e_2~ x~ y)
                                    &&= \lambda x.\lambda y.\lambda z.~ e_1~ (x~ (y~ z))
\\&=
\lambda x.~ B~ e_1~ (B~ e_2~ x)
                                    &&= \lambda x.\lambda y.\lambda z.~ e_1~ (B~ x~ y~ z)
\\&=
B~ (B~ e_1)~ (B~ e_2)
                                    &&= \lambda x.\lambda y.~ B~ e_1~ (B~ x~ y)
\\&
                                    &&= \lambda x.~ B~ (B~ e_1)~ (B~ x)
\\&                                    
                                    &&= B~ (B~ (B~ e_1))~ B
\end{align*}
where the $\alpha$-renaming is implicitly used.
\end{proof}

Equation~(B2) has been employed by Statman~\cite{Statman10type}
to show that no $B\omega$-term can be a fixed-point combinator
where $\omega=\lambda x.x~x$.
This equation exposes an interesting feature of the $B$ combinator.
Write equation~(B2) as
\begin{align}
B~ (e_1 \circ e_2) = (B~ e_1) \circ (B~ e_2)
\tag{B2'}
\labeqn{Bo-distr}
\end{align}
by replacing every $B$ combinator with $\circ$ infix operator
if it has exactly two arguments.
The equation is an distributive law of $B$ over $\circ$,
which will be used to obtain the canonical representation of $B$-terms.
Equation~(B3) is also used for the same purpose
as the form of
\begin{align}
B \circ (B~ e_1) = (B~ (B~ e_1)) \circ B\text.
\tag{B3'}
\labeqn{Bo-push}
\end{align}

One may expect a natural equation
\begin{equation}
B~ e_1~ (B~ e_2~ e_3) = B~ (B~ e_1~ e_2)~ e_3
\labeqn{B-assoc}
\end{equation}
which represents associativity of function composition,
i.e.,~$e_1\circ (e_2\circ e_3) = (e_1\circ e_2)\circ e_3$
This is shown with equations~(B1) and (B2)
by
\begin{align*}
B~ e_1~ (B~ e_2~ e_3) = B~ (B~ e_1)~ (B~ e_2)~ e_3 = B~ (B~ e_1~ e_2)~ e_3
\text.
\end{align*}

\subsection{Canonical representation of $B$-terms}
\labsec{canonical}
To decide equality between two $B$-terms,
it does not suffices to compute their normal forms
under the definition of $B$, $B~x~y~z\to x~(y~z)$.
This is because
two distinct normal forms may be equal up to $\beta\eta$-equivalence,
e.g., $B~B~(B~B)$ and $B~(B~(B~B))~B$.
We introduce a canonical representation of $B$-terms,
which makes it easy to check equivalence of $B$-terms.
We will finally find that for any $B$-term $e$
there exists a unique finite non-empty weakly-decreasing sequence
of non-negative integers
$n_1\geq n_2\geq\dots\geq n_k$ 
such that $e$ is 
equivalent
to $(B^{n_1} B) \circ (B^{n_2} B)\circ \dots \circ (B^{n_k} B)$.
Ignoring the inequality condition gives 
\emph{polynomials} introduced by Statman~\cite{Statman10type}.
We will use \emph{decreasing polynomials} for our canonical representation
as presented later.

First, we explain how its canonical form is obtained from every $B$-term.
We only need to consider $B$-terms 
in which every $B$ has at most two arguments.
One can easily reduce the arguments of $B$ to less than three
by repeatedly rewriting occurrences of $B~e_1~e_2~e_3~e_4~\dots~e_n$ into
$e_1~(e_2~e_3)~e_4~\dots~e_n$.
The rewriting procedure always terminates because it reduces the number of $B$.
Thus, every $B$-term in $\CL{B}$ is equivalent to a $B$-term built
by the syntax
\begin{align}
e &::= B ~\mid~ B ~ e ~\mid~ e ~ \circ ~ e
\labeqn{syntaxBo}
\end{align}
where $e_1\circ e_2$ denotes $B~e_1~e_2$.
We prefer to use the infix operator $\circ$ instead of $B$ that has two arguments
because associativity of $B$, that is, 
$B~e_1~(B~e_2~e_3) = B~(B~e_1~e_2)~e_3$
can be implicitly assumed.
This simplifies the further discussion on $B$-terms.
We will deal with only $B$-terms in syntax~\refeqn{syntaxBo} from now on.
The $\circ$ operator has a lower precedence than application in this paper.
Terms $B~B\circ B$ and $B\circ B~B$ represent $(B~B)\circ B$ and $B\circ(B~B)$, respectively.

The syntactic restriction by \refeqn{syntaxBo} does not suffice
to proffer a canonical representation of $B$-terms.
There are many pairs of $B$-terms which are equivalent
even in the form of \refeqn{syntaxBo}.
For example, $B\circ B~ B = B~(B~B) \circ B$ holds according to 
equation~\refeqn{Bo-push}.

A \emph{polynomial form of $B$-terms} is obtained
by putting a restriction to the syntax
so that no $B$ combinator occurs outside of the $\circ$ operator
while syntax~\refeqn{syntaxBo} allows
the $B$ combinators and the $\circ$ operators to occur
in arbitrary position.
The restricted syntax is given as
\begin{align*}
e &::= e_B ~\mid~ e \circ e 
\\ 
e_B &::= B ~\mid B~ e_B
\end{align*}
where terms in $e_B$ have a form of
$\underbrace{B (\dots(B}_{n} B)\dots)$, that is $B^n B$,
called \emph{monomial}.
The syntax can be simply rewritten into
$e ::= B^n B ~\mid~ e \circ e$,
which is called \emph{polynomial}.

\begin{definition}\rm
A $B$-term $B^n B$ is called \emph{monomial}.
A \emph{polynomial} 
is given as the form of
\begin{align*}
(B^{n_1} B) \circ (B^{n_2} B) \circ\dots\circ (B^{n_k} B)
\end{align*}
where $k>0$ and $n_1,\dots, n_k\geq0$ are integers.
In particular,
a polynomial is called \emph{decreasing}
when $n_1\geq n_2\geq\dots\geq n_k$.
The \emph{length} of a polynomial $P$
is defined by adding 1 to the number of $\circ$ in $P$.
The numbers $n_1, n_2,\dots,n_k$ are called \emph{degree}.
\end{definition}

In the rest of this subsection,
we prove that for any $B$-term $e$ there exists
a unique decreasing polynomial equivalent to $e$.
First, we show that $e$ has an equivalent polynomial.

\begin{lemma}[\cite{Statman10type}]
\lablem{H-exist}
For any $B$-term $e$,
there exists a polynomial to $e$.
\end{lemma}
\begin{proof}
We prove the statement by induction on structure of $e$.
In the case of $e\equiv B$, the term itself is polynomial.
In the case of $e\equiv B~e_1$,
assume that $e_1$ has equivalent polynomial
$(B^{n_1} B)\circ(B^{n_2} B)\circ\dots\circ(B^{n_k} B)$.
Repeatedly applying equation~\refeqn{Bo-distr} to $B~e_1$,
we obtain a polynomial equivalent to $B~e_1$ 
as $(B^{n_1+1} B) \circ (B^{n_2+1} B) \circ\dots\circ (B^{n_k+1} B)$.
In the case of $e\equiv e_1\circ e_2$,
assume that $e_1$ and $e_2$ have equivalent polynomials
$P_1$ and $P_2$, respectively.
A polynomial equivalent to $e$ is given
by $P_1\circ P_2$.
\end{proof}


Next we show that for any polynomial $P$
there exists a decreasing list equivalent to $P$.
A key equation of the proof is
\begin{align}
(B^m B) \circ (B^n B) &=
(B^{n+1} B) \circ (B^m B)
\quad
\text{when~$m<n$,}
\labeqn{B-swap}
\end{align}
which is shown by
\begin{align*}
(B^m B) \circ (B^n B)
&=
B^m (B \circ(B^{n-m} B))
\\&=
B^m (B \circ (B~(B^{n-m-1} B)))
\\&=
B^m ((B (B (B^{n-m-1} B))) \circ B)
\\&=
(B^{n+1} B) \circ (B^m B)
\end{align*}
using equations~\refeqn{Bo-distr} and~\refeqn{Bo-push}.

\begin{lemma}
\lablem{H-decr}
Any polynomial $P$ has an equivalent decreasing polynomial $P'$
such that
\begin{itemize}
\item the length of $P$ and $P'$ are equal, and
\item the lowest degrees of $P$ and $P'$ are equal.
\end{itemize}
\end{lemma}
\begin{proof}
We prove the statement by induction on the length of $P$.
When the length is 1, that is, $P$ includes no $\circ$ operator,
$P$ itself is decreasing and the statement holds.
When the length of $P$ is $k>1$,
take $P_1$ such that $P \equiv P_1 \circ (B^n B)$.
From the induction hypothesis,
there exists a decreasing polynomial
$P_1'\equiv(B^{n_1} B)\circ(B^{n_2} B)\circ\dots\circ (B^{n_{k-1}} B)$
equivalent to $P_1$,
and the lowest degree of $P_1$ is $n_{k-1}$.
If $n_{k-1}\geq n$, then $P'\equiv P_1'\circ (B^n~ B)$ is decreasing
and equivalent to $P$.
Since the lowest degrees of $P$ and $P'$ are $n$,
the statement holds.
If $n_{k-1} < n$,
$P$ is equivalent to
\begin{align*}
\lefteqn{
(B^{n_1}~ B) \circ
\dots
\circ (B^{n_{k-2}} B) \circ (B^{n_{k-1}} B) \circ(B^n B)
}
\\&=
(B^{n_1} B) \circ
\dots
\circ (B^{n_{k-2}} B) \circ (B^{n+1} B) \circ (B^{n_{k-1}} B)
\end{align*}
due to equation~\refeqn{B-swap}.
Putting the last term as $P_2\circ(B^{n_{k-1}} B)$,
the length of $P_2$ is $k-1$
and the lowest degree of $P_2$ is greater than or equal to $n_{k-1}$.
From the induction hypothesis,
$P_2$ has an equivalent decreasing polynomial $P_2'$
of length $k-1$
and the lowest degree of $P_2'$ greater than or equal to $n_{k-1}$.
Thereby we obtain
a decreasing polynomial $P_2'\circ(B^{n_{k-1}} B)$
equivalent to $P$
and the statement holds.
\end{proof}

\begin{example}\rm
Consider a $B$-term $e=B~(B~B~B)~(B~B)~B$.
First, applying equation~(B1),
\begin{align*}
e &= B~(B~B~B)~(B~B)~(B~B)
= B~B~B~(B~B~(B~B))
= B~(B~(B~B~(B~B)))
\end{align*}
so that every $B$ has at most two arguments.
Then replace each $B$ to the infix $\circ$ operator
if it has two arguments and obtain $B~(B~(B\circ(B~B)))$
Applying equation~\refeqn{Bo-distr}, we have
\begin{align*}
B~(B~(B\circ(B~B)))
&=
B~((B~B)\circ(B~(B~B)))
\\&
=
(B~(B~B))\circ(B~(B~(B~B)))
\\&= (B^2 B)\circ(B^3 B)\text.
\end{align*}
Applying equation~\refeqn{B-swap}, 
we obtain the decreasing polynomial
$(B^4 B)\circ(B^2 B)$ equivalent to $e$.
\end{example}

Every $B$-term
has at least one equivalent decreasing polynomial as shown so far.
To conclude this subsection,
we show the uniqueness of decreasing polynomial equivalent to any $B$-term,
that is,
every $B$-term $e$ has no two distinct decreasing polynomials equivalent to $e$.

The proof is based on the idea that
$B$-terms correspond to unlabeled binary trees.
In every $\lambda$-term in $\CL{B}$,
all variables are referred exactly once (linear)
in the order they it was introduced (ordered).
More precisely, this fact is formalized as follows.
\begin{lemma}
\lablem{bform}
Every $\lambda$-term in $\CL{B}$ is $\beta\eta$-equivalent to a $\lambda$-term
of the form
$\lambda x_1.\lambda x_2.\dots.\lambda x_k.~M$ with some $k>2$ 
where $M$ is built by putting parentheses 
to appropriate positions in the sequence $x_1~ x_2~ \dots~ x_k$.
\end{lemma}
\begin{proof}
This can be proved by induction on structure of terms in $\CL{B}$.
\end{proof}

This lemma implies that 
every $\lambda$-term in $\CL{B}$ is characterized 
by an unlabeled binary tree.
A $\lambda$-term in $\CL{B}$ is constructed for any unlabeled binary tree
by putting a variable to each leaf in the order of $x_1, x_2,\dots$
and enclosing it with $k$-fold lambda abstraction
$\lambda x_1.\lambda x_2.\dots.\lambda x_k.~[~]$
where $k$ is a number of leaves of the binary tree.
Let us use the notation $\x$ for a leaf
and $\<t_1,t_2>$ for a tree
with left subtree $t_1$ and right subtree $t_2$.
For example,
$B$-terms $B=\lambda x.\lambda y.\lambda z.~x~(y~z)$ and
$B~B=\lambda x.\lambda y.\lambda z.\lambda w.~x~y~(z~w)$ are
represented by
$\<\x,\<\x,\x>>$ and
$\<\<\x,\x>,\<\x,\x>>$, respectively.

We will present an algorithm
for constructing the corresponding decreasing polynomial from a given binary tree.
First let us define a function $\nodes_i$ with integer $i$
which maps binary trees to lists of integers:
\begin{align*}
\nodes_i(\x) &= [~]
&
\nodes_i(\<t_1,t_2>) &=
\nodes_{i+\size{t_1}}(t_2) \concat \nodes_i(t_1) \concat [i]
\end{align*}
where $\concat$ concatenates two lists
and $\size{t}$ denotes a number of leaves.
For example,
$\nodes_0(\<\<\x,\x>,\<\x,\x>>) = [2,0,0]$ and
$\nodes_1(\<\<\x,\<\x,\x>>,\<\x,\<\x,\x>>>) = [4,4,2,1,1]$.
Informally, the $\nodes_i$ function returns a list of integers
which is obtained by labeling both leaves and nodes
in the following steps.
First each leaf of a given tree is labeled by $i,i+1,i+2,\dots$
in left-to-right order.
Then each binary node of the tree is labeled by
the same label as its leftmost descendant leaf.
The $\nodes_i$ functions returns a list of only node labels
in decreasing order.
The length of the list equals
the number of nodes, that is, 
smaller by one than the number of variables in the $\lambda$-term.

We define a function $\nodes$ which takes a binary tree $t$ and
returns a list of non-negative integers in $\nodes_{-1}(t)$,
that is, the list obtained by excluding trailing all $-1$'s in $\nodes_{-1}(t)$.
Note that by excluding the label $-1$'s
it may happen to be $\nodes(t) = \nodes(t')$ for
two distinct binary trees $t$ and $t'$
even though the $\nodes_i$ function is injective.
However, those binary trees $t$ and $t'$ must be `$\eta$-equivalent'
in terms of the corresponding $\lambda$-terms.

The following lemma claims that the $\nodes$ function computes
a list of degrees of a decreasing polynomial corresponding to a given $\lambda$-term.

\begin{lemma}
\lablem{H-uniq}
A decreasing polynomial $(B^{n_1} B)\circ(B^{n_2} B)\circ\dots\circ(B^{n_k} B)$
is $\beta\eta$-equivalent to a $\lambda$-term $e\in\CL{B}$
corresponding a binary tree $t$ such that
$\nodes(t) = [n_1,n_2,\dots,n_k]$.
\end{lemma}
\begin{proof}
We prove the statement by induction on the length of the polynomial $P$.

When $P\equiv B^{n} B$ with $n\geq0$,
it is found to be equivalent to the $\lambda$-term
\begin{equation*}
\lambda x_1.
\lambda x_2.
\lambda x_3.
\dots.
\lambda x_{n+1}.\lambda x_{n+2}.
\lambda x_{n+3}.~ x_1~
x_2~ 
x_3~ 
\dots~ x_{n+1}~ (x_{n+2}~ x_{n+3})
\end{equation*}
by induction on $n$.
This $\lambda$-term corresponds to a binary tree
$t = \<\poly{\text{$n$ leaves}},\<\x,\x>>$.
Then we have $\nodes(t) = [n]$ holds
from $\nodes_{-1}(t) = [n,\underbrace{-1,-1,\dots,-1}_{n+1}]$.

When $P\equiv P'\circ(B^{n} B)$ with
$P'\equiv(B^{n_1} B)\circ\dots\circ(B^{n_{k}} B)$,
$k\geq 1$ and $n_1\geq\dots\geq n_{k}\geq n\geq0$,
there exists a $\lambda$-term $\beta\eta$-equivalent to $P'$
corresponding a binary tree $t'$ such that $\nodes(t')=[n_1,\dots,n_{k}]$
from the induction hypothesis.
The binary tree $t'$ must have
the form of $\<\<\poly{\text{$n_k$ leaves}},t_1>,\dots,t_m>$
with $m\geq1$ and some trees $t_1,\dots,t_m$,
otherwise $\nodes(t')$ would contain an integer smaller than $n_k$.
From the definition of $\nodes$ and $\nodes_i$, we have
\begin{align}
\nodes(t') &= \nodes_{s_m}(t_m) \concat \dots \concat \nodes_{s_1}(t_1)
\labeqn{nodes_t'}
\end{align}
where $s_j = n_k + 1 + \sum_{i=1}^{j-1}\size{t_i}$.
Additionally, the structure of $t'$ implies
$
P' =
\lambda x_1.\dots.\lambda x_{l}.~\allowbreak
x_1~x_2~\dots~x_{n_k+1}~e_1\dots e_m
$
where $e_i$ corresponds to a binary tree $t_i$ for $i=1,\dots,m$.
From $B^{n}~B=
\lambda y_1.\dots.\lambda y_{n+3}.~\allowbreak
y_1~ y_2\dots y_{n+1}~ (y_{n+2}~ y_{n+3})$,
we compute a $\lambda$-term $\beta\eta$-equivalent to $P\equiv P'\circ(B^n B)$ by
\begin{align*}
P
&=
\lambda x.~ P' (B^{n} B~ x)
\\&=
\lambda x.~
(\lambda x_1.\dots.\lambda x_{l}.~
x_1~x_2\dots x_{n_k+1}~e_1~\dots~e_{m})
\\&\qquad\qquad
(\lambda y_2.\dots.\lambda y_{n+3}.~
x~ y_2\dots y_{n+1}~ (y_{n+2}~ y_{n+3}))
\\&=
\lambda x.
\lambda x_2.\dots.\lambda x_{l}.~
(\lambda y_2.\dots.\lambda y_{n+3}.~
x~ y_2\dots y_{n+1}~ (y_{n+2}~ y_{n+3}))
~x_2\dots x_{n_k+1}~e_1\dots e_{m}
\\&=
\lambda x.
\lambda x_2.\dots.\lambda x_{l}.~
\\&\qquad
(\lambda y_{n+1}.\lambda y_{n+2}.\lambda y_{n+3}.~
x~ x_2\dots x_n~y_{n+1}~(y_{n+2}~y_{n+3}))~
x_{n+1}\dots x_{n_k+1}~e_1\dots e_{m}
\end{align*}
where $n_k\geq n$ is taken into account.
We split into four cases:
(i) $n_k = n$ and $m=1$, 
(ii) $n_k = n$ and $m>1$, 
(iii) $n_k = n+1$, and
(iv) $n_k > n+1$.
In the case~(i) where $n_k = n$ and $m=1$, we have
\begin{align*}
P &= 
\lambda x.\lambda x_2.\dots.\lambda x_{l}.\lambda y_{n+3}.~
x~ x_2\dots x_n~x_{n+1}~(e_1~y_{n+3})\text.
\end{align*}
whose corresponding binary tree $t$ is 
$\<\poly{\text{$n$ leaves}},\<t_1,\x>>$.
From equation~\refeqn{nodes_t'},
$\nodes(t) = \nodes_{n+1}(t_1)\concat [n+1]
= \nodes(t')\concat[n+1] = [n_1,\dots,n_k,n+1]$, thus the statement holds.
In the case~(ii) where $n_k = n$ and $m>1$, we have
\begin{align*}
P &= 
\lambda x.\lambda x_2.\dots.\lambda x_{l}.~
x~ x_2\dots x_n~x_{n+1}~(e_1~e_2)~e_3\dots e_{m}\text.
\end{align*}
whose corresponding binary tree $t$ is 
$\<\<\poly{\text{$n$ leaves}},\<t_1,t_2>,t_3>,\dots,t_m>$.
Hence, $\nodes(t)=\nodes(t')\concat[n+1]$ holds again from equation~\refeqn{nodes_t'}.
In the case~(iii) where $n_k = n+1$, we have
\begin{align*}
&P = 
\lambda x.\lambda x_2.\dots.\lambda x_{l}.~
x~ x_2\dots x_n~x_{n+1}~(x_{n+2}~e_1)~e_2\dots e_{m}\text{, or}
\end{align*}
whose corresponding binary tree $t$ is 
$\<\<\poly{\text{$n$ leaves}},\<\x,t_1>,t_2>,\dots,t_m>$.
Hence, $\nodes(t)=\nodes(t')\concat[n+1]$ holds from equation~\refeqn{nodes_t'}.
In the case~(iv) where $n_k\geq n+2$, we have
\begin{align*}
P = 
\lambda x.\lambda x_2.\dots.\lambda x_{l}.~
x~ x_2\dots x_n~x_{n+1}~(x_{n+2}~x_{n+3})~\dots~e_1\dots e_{m}\text,
\end{align*}
whose corresponding binary tree $t$ is 
$\<\<\poly{\text{$n$ leaves}},\<\x,\x>,\dots,t_1>,\dots,t_m>$.
Hence, $\nodes(t)=\nodes(t')\concat[n+1]$ holds from equation~\refeqn{nodes_t'}.
\end{proof}

\begin{example}\rm
A $\lambda$-term 
$\lambda x_1.\lambda x_2.\lambda x_3.\lambda x_4.\lambda x_5.\lambda x_6.
\lambda x_7.\lambda x_8.~
x_1~ (x_2~ x_3)~ (x_4~ x_5~ x_6~ (x_7~ x_8))$
is $\beta\eta$-equivalent to 
$(B^{5}~ B)\circ(B^{2}~ B)\circ(B^{2}~ B)\circ(B^{2}~ B)\circ(B^{0}~ B)$
because
its corresponding binary tree 
$t=\<\<\<\x,\x>>,\<\<\<\x,\x>,\x>,\<\x,\x>>>$
satisfies
$\nodes(t) = [5,2,2,2,0]$.
\end{example}

The previous lemmas immediately conclude the uniqueness of 
of decreasing polynomials for a $B$-term shown in the following theorem.


%
\begin{theorem}
\labthr{canonical}
Every $B$-term $e$ has a unique decreasing polynomial.
\end{theorem}
\begin{proof}
For any given $B$-term $e$, 
we can find a decreasing polynomial for $e$
from \reflem{H-exist} and \reflem{H-decr}.
Since no other decreasing polynomial an be equivalent to $e$
from \reflem{H-uniq}, the present statement holds.
\end{proof}

This theorem implies that
the decreasing polynomial of $B$-terms can be used 
as their \emph{canonical representation},
which is effectively derived as shown in
\reflem{H-exist} and \reflem{H-decr}.

As a corollary of the theorem,
we can show the only-if statement of \refthr{eqB},
which corresponds to the completeness of the equation system.

\begin{proof}[\proofname ~of only-if-part of \refthr{eqB}]
Let $e_1$ and $e_2$ be equivalent $B$-terms,
that is, their $\lambda$-terms are $\beta\eta$-equivalent.
From \refthr{canonical},
their decreasing polynomials are the same.
Since the decreasing polynomial is derived from $e_1$ and $e_2$
by equations~(B1), (B2), and (B3)
according to the proofs of \reflem{H-exist} and \reflem{H-decr},
equivalence between $e_1$ and $e_2$ is also derived from these equations.
\end{proof}


\section{Results on the $\rho$-property of $B$-terms}
\labsec{results}
We investigate the $\rho$-property of concrete $B$-terms,
some of which have the property and others do not.
For $B$-terms having the $\rho$-property,
we introduce an efficient implementation 
to compute the entry point and the size of the cycle.
For $B$-terms not having the $\rho$-property,
we give a proof why they do not have.
%
%
%

\subsection{$B$-terms having the $\rho$-property}
As shown in~\refsec{prelim},
we can check that
$B$-terms equivalent to $B^n B$ with $n\leq 6$
have the $\rho$-property
by computing $\sapp{(B^n B)}{i}$ for each $i$.
However, it is not easy to check it by computer
without an efficient implementation
because we should compute all $\sapp{(B^6 B)}{i}$ with
$i\leq 2980054085040~ (= 2641033883877+339020201163)$
to know $\rho(B^6 B)=(2641033883877,339020201163)$.
A naive implementation
which computes terms of $\sapp{(B^6 B)}{i}$ for all $i$
and stores all of them
has no hope to detect the $\rho$-property.

We introduce an efficient procedure
to find the $\rho$-property of $B$-terms
which can successfully compute $\rho(B^6 B)$.
The procedure is based on two orthogonal ideas,
Floyd's cycle-finding algorithm~\cite{Knuth69book}
and
an efficient right application algorithm over decreasing polynomials
presented in~\refsec{canonical}.

The first idea, Floyd's cycle-finding algorithm
(also called the tortoise and the hare algorithm),
enables to detect the cycle with a constant memory usage,
that is, the history of all terms $\sapp{X}{i}$ does not
need to be stored to check the $\rho$-property of the $X$ combinator.
The key of this algorithm is the fact that
there are two distinct integers $i$ and $j$ with $\sapp{X}{i}=\sapp{X}{j}$
if and only if 
there are an integer $m$ with $\sapp{X}{m}=\sapp{X}{2m}$,
where the latter requires to compare $\sapp{X}{i}$ and $\sapp{X}{2i}$
from smaller $i$ and store only these two terms
for the next comparison between
$\sapp{X}{i+1}=\sapp{X}{i}X$
and $\sapp{X}{2i+2}=\sapp{X}{2i}XX$
when $\sapp{X}{i}\neq\sapp{X}{2i}$.
The following procedure computes the entry point and the size of the cycle
if $X$ has the $\rho$-property.
\begin{enumerate}
\item Find the smallest $m$ such that $\sapp{X}{m}=\sapp{X}{2m}$.
\item Find the smallest $k$ such that $\sapp{X}{k}=\sapp{X}{m+k}$.
\item Find the smallest $0<c\leq k$ such that $\sapp{X}{m}=\sapp{X}{m+c}$.
If not found, put $c=m$.
\end{enumerate}
After this procedure, we find $\rho(X)=(k,c)$.
The third step can be simultaneously run during the second one.
See~\cite[exercise 3.1.6]{Knuth69book} for the detail.
One could use slightly more (possibly) efficient algorithms
by Brent~\cite{Brent80bit} and Gosper~\cite[item 132]{Beeler72mit}
for cycle detection.

Efficient cycle-finding algorithms do not suffice to compute $\rho(B^6 B)$.
Only with the idea above
running on a laptop (1.7 GHz Intel Core i7 / 8GB of memory),
it takes about 2 hours even for $\rho(B^5 B)$
fails to compute $\rho(B^6 B)$ due to out of memory.

The second idea
enables to efficiently compute $\sapp{X}{i+1}$ from $\sapp{X}{i}$
for $B$-terms $X$.
The key of this algorithm is to use the canonical representation
of $\sapp{X}{i}$, that is a decreasing polynomial,
and directly compute the canonical representation of $\sapp{X}{i+1}$
from that of $\sapp{X}{i}$.
Our implementation based on both ideas 
can compute $\rho(B^5 B)$ and $\rho(B^6 B)$
in 10 minutes and 59 days (!), respectively.


%
For two given decreasing polynomials $P_1$ and $P_2$,
we show how a decreasing polynomial $P$ equivalent to $(P_1~P_2)$
can be obtained.
The method is based on the following lemma
about application of one $B$-term to another $B$-term.
%
%
\begin{lemma}
\lablem{comp-app}
For $B$-terms $e_1$ and $e_2$,
there exists $k\geq0$ such that
$e_1\circ(B~ e_2) = B~ (e_1~ e_2)\circ B^k$.
\end{lemma}
\begin{proof}
Let $P_1$ be a decreasing polynomial equivalent to $e_1$.
We prove the statement by case analysis on the maximum degree in $P_1$.
When the maximum degree is 0, 
we can take $k'\geq1$ such that
$P_1\equiv\underbrace{B\circ\dots\circ B}_{k'}=B^{k'}$.
Then,
\begin{align*}
e_1\circ(B~ e_2)&= 
\underbrace{B\circ\dots\circ B}_{k'}
\circ(B~ e_2) 
= (B^{k'+1} e_2)\circ\underbrace{B\circ\dots\circ B}_{k'}
= B~ (e_1~ e_2)\circ B^{k'}
\end{align*}
where equation~\refeqn{Bo-push} is used $k'$ times in the second equation.
Therefore the statement holds by taking $k=k'$.
When the maximum degree is greater than 0,
we can take a decreasing polynomial $P'$ for a $B$-term and $k'\geq0$ such that
$P_1=(B~ P')\circ \underbrace{B\circ\dots\circ B}_{k'}=
(B~ P')\circ B^{k'}$ due to equation~\refeqn{Bo-distr}.
Then,
\begin{align*}
e_1\circ(B~ e_2)&= 
(B~ P')\circ\underbrace{B\circ\dots\circ B}_{k'}\circ(B~ e_2)
\\&=
(B~ P')\circ(B^{k'+1} e_2)\circ\underbrace{B\circ\dots\circ B}_{k'}
\\&=
B~ (P'\circ(B^{k'} e_2))\circ B^{k'}
\\&=
B~ (B~P'~(B^{k'} e_2))\circ B^{k'}
\\&=
B~ (P_1~ e_2)\circ B^{k'}
\\&=
B~ (e_1~ e_2)\circ B^{k'}\text.
\end{align*}
Therefore, the statement holds by taking $k=k'$.
\end{proof}

This lemma indicates that,
from two decreasing polynomials $P_1$ and $P_2$,
a decreasing polynomial $P$ equivalent to $(P_1~P_2)$
can be obtained by the following steps:
\begin{enumerate}
\item Build $P_2'$ by raising each degree of $P_2$ by 1,
\ie,~when $P_2\equiv (B^{n_1} B)\circ\dots\circ(B^{n_l} B)$,
$P_2' \equiv(B^{n_1+1} B)\circ\dots\circ(B^{n_l+1} B)$.
\item Find a decreasing polynomial $P_{12}$
corresponding to $P_1\circ P_2'$ by equation~\refeqn{Bo-push};
\item Obtain $P$ by lowering each degree of $P_{12}$ after
eliminating the trailing 0-degree units,
i.e.,~when $P_{12}\equiv (B^{n_1} B)\circ\dots\circ(B^{n_l} B)
\circ(B^{0} B)\circ\dots\circ(B^{0} B)$ with $n_1\geq\dots\geq n_l>0$,
$P\equiv(B^{n_1-1} B)\circ\dots\circ(B^{n_l-1} B)$.
\end{enumerate}
In the first step,
a decreasing polynomial $P_2'$ equivalent to $B~P_2$
is obtained.
The second step yields
a decreasing polynomial $P_{12}$
for $P_1\circ P_2'=P_1\circ(B~P_2)$.
Since $P_1$ and $P_2$ are decreasing,
it is easy to find $P_{12}$
by repetitive application of equation~\refeqn{Bo-push}
for each unit of $P_2'$,
{\`a}~la insertion operation in insertion sort.
In the final step,
a polynomial $P$ that satisfies
$(B~ P)\circ B^{k} = P_{12}$ with some $k$ is obtained.
From \reflem{comp-app} and the uniqueness of decreasing polynomials, 
$P$ is equivalent to $(P_1~P_2)$.

\begin{example}\rm
Let $P_1$ and $P_2$ be decreasing polynomials given as
$P_1 = (B^4 B)\circ(B^1 B)\circ(B^0 B)$
and
$P_2 = (B^2 B)\circ(B^0 B)$.
Then a decreasing polynomial $P$ equivalent to $(P_1~P_2)$
is obtained by three steps:
\begin{enumerate}
\item Raise each degree of $P_2$ to get 
$P_2'=(B^3 B)\circ(B^1 B)$.
\item Find a decreasing polynomial by
\begin{align*}
P_1\circ P_2'&=
(B^4 B)\circ(B^1 B)\circ(B^0 B)\circ(B^3 B)\circ(B^1 B)
\\&=
(B^4 B)\circ(B^1 B)\circ(B^4 B)\circ(B^0 B)\circ(B^1 B)
\\&=
(B^4 B)\circ(B^5 B)\circ(B^1 B)\circ(B^0 B)\circ(B^1 B)
\\&=
(B^6 B)\circ(B^4 B)\circ(B^1 B)\circ(B^0 B)\circ(B^1 B)
\\&=
(B^6 B)\circ(B^4 B)\circ(B^1 B)\circ(B^2 B)\circ(B^0 B)
\\&=
(B^6 B)\circ(B^4 B)\circ(B^3 B)\circ(B^1 B)\circ(B^0 B)
\end{align*}
where equation~\refeqn{Bo-push} is applied in each.
\item
By lowering each degree after removing trailing $(B^0 B)$'s,
\begin{align*}
P&\equiv(B^5 B)\circ(B^3 B)\circ(B^2 B)\circ(B^0 B)
\end{align*}
is obtained.
\end{enumerate}
\end{example}

The implementation based on the right application over decreasing polynomials e is available at
\url{https://github.com/ksk/Rho}.
Note that the program does not terminate for the combinator
which does not have the $\rho$-property.
It will not help to decide if a combinator has the $\rho$-property.
One might observe how the terms grow by repetitive right applications
thorough running the program, though.

\subsection{$B$-terms not having the $\rho$-property}
%
%

We prove that the $B$-terms $(B^k B)^{(k+2)n}$ ($k \ge 0,\ n > 0$)
do not have the $\rho$-property.
In particular, we show that the number of variables 
in the $\beta\eta$-normal form of $((B^k B)^{(k+2)n})_{(i)}$ is monotonically non-decreasing
and that it implies the anti-$\rho$-property.
Additionally, after proving that,
we consider a sufficient condition not to have the $\rho$-property through the monotonicity.

First we introduce some notation.
Suppose that the $\beta\eta$-normal form of a $B$-term $X$ is given by $\lambda x_1.\dots\lambda x_n.~\allowbreak x_1~ e_1~ \cdots~ e_k$ for some terms $e_1, \dots, e_k$.
Then we define $l(X) = n$ (the number of variables), $a(X) = k$ (the number of arguments of $x_1$), and $N_i(X) = e_i$ for $i = 1, \dots, k$.
Let $X'$ be another $B$-term and suppose its $\beta\eta$-normal form is given by $\lambda x'_1. \dots \lambda x'_{n'}.~ e'$,
We can see $X~ X' = (\lambda x_1. \ldots \lambda x_n.~ x_1~ e_1~ \cdots~ e_k)~ X' = \lambda x_2. \dots \lambda x_n.~ X'~ e_1~ \cdots~ e_k$
and from \reflem{bform}, its $\beta\eta$-normal form is
\[
  \begin{cases}
    \lambda x_2.\dots \lambda x_n. \lambda x'_{k+1}. \dots \lambda x'_{n'}.~ [e_1, \ldots, e_k]e' & (k \le n')\\
    \lambda x_2.\dots \lambda x_n.~ [e_1, \dots, e_{n'}]e'~ e_{n'+1} ~ \cdots ~ e_k & (\text{otherwise}).
  \end{cases}
\]
Here $[e_1,\dots,e_k]e'$ is the term which is obtained by substituting $e_1, \dots, e_k$ to the variables $x'_1,\dots,x'_k$ in $e'$ in order.

By simple computation with this fact, we get the following lemma:

\begin{lemma}
\lablem{asymp}
Let $X$ and $X'$ be $B$-terms. Then
\begin{align*}
  l(X~X') &= l(X) -1 + \max\{l(X') - a(X), \ 0\}\\
  a(X~X') &= a(X') + a(N_1(X)) + \max\{a(X) - l(X'),\  0\}\\
  N_1(X~X') &=
  \begin{cases}
    [N_2(X),\ldots,N_m(X)]N_1(X') & (\text{if $N_1(X)$ is a variable})\\
    N_1(N_1(X)) & (\text{otherwise})
  \end{cases}
\end{align*}
where $m = \min\{l(N_1(X')),\  a(X)\}$.
\end{lemma}

The $\beta\eta$-normal form of $(B^k B)^{(k+2)n}$ is given by
\[
  \lambda x_1.\ldots\lambda x_{k+(k+2)n+2}.~ x_1~ x_2~ \cdots~ x_{k+1}~ (x_{k+2}~ x_{k+3~} \cdots~ x_{k+(k+2)n+2}).
\]
This is deduced from \reflem{H-uniq} since
the binary tree corresponding to the above $\lambda$-term is
$t = 
\allowbreak
\<\<\ldots\langle\langle\underbrace{\x,\x\rangle, \x\rangle,\ldots,\x}_{k+1}>,
\allowbreak
 \<\ldots\langle\langle\x,\underbrace{\x\rangle,\x\rangle,\ldots,\x}_{(k+2)n}>>$
and $\mathcal{L}(t) = [\underbrace{k,\ldots,k}_{(k+2)n}]$.
Especially, we get $l((B^k B)^{(k+2)n}) = k + (k+2)n + 2$.
In this section, we write $\<\x,\x,\x,\ldots,\x>$ for $\<\ldots\<\<\x,\x>,\x>,\ldots,\x>$ and identify $B$-terms with their corresponding binary trees.

To describe properties of $(B^k B)^{(k+2)n}$,
we introduce a set $T_{k,n}$ which is closed under right application of $(B^k B)^{(k+2)n}$,
that is, $T_{k,n}$ satisfies that ``if $X \in T_{k,n}$ then $X~ (B^k B)^{(k+2)n} \in T_{k,n}$ holds''.
First we inductively define a set of terms $T'_{k,n}$ as follows:
\begin{enumerate}
\item $\x \in T'_{k,n}$
\item $\langle\x,~ s_1,~ \ldots,~ s_{(k+2)n}\rangle \in T'_{k,n}$ if
$s_i=\x$ for a multiple $i$ of $k+2$
and $s_i\in T'_{k,n}$ for the others.
\end{enumerate}
Then we define $T_{k,n}$ by
$T_{k,n} = \left\{ \langle t_0,~ t_1,~ \ldots,~ t_{k+1}\rangle \ \middle| \ t_0,\ t_1, \ldots, t_{k+1} \in T'_{k,n} \right\}$.
It is obvious that $(B^k B)^{(k+2)n} \in T_{k,n}$.
Now we shall prove that $T_{k,n}$ is closed under right application of $(B^k B)^{(k+2)n}$.
\begin{lemma}
\lablem{closed}
If $X \in T_{k,n}$ then $X~ (B^k B)^{(k+2)n} \in T_{k,n}$.
\end{lemma}
\begin{proof}
From the definition of $T_{k,n}$, 
if $X \in T_{k,n}$ then $X$ can be written
in the form $\langle t_0,\ t_1,\ \dots,\ t_{k+1}\rangle$
for some $t_0, \ldots, t_{k+1} \in T'_{k,n}$.
In the case where $t_0 = \x$, we have $X~ (B^k B)^{(k+2)n} = \langle t_1,~\ldots,~ t_{k+1},~ \langle\underbrace{\x,~ \ldots,~ \x}_{(k+2)n}\rangle\rangle \in T_{k,n}$.
In the case where $t_0$ has the form of 2 in the definition of $T'_{k,n}$,
then $X = \langle\x,~ s_1,~ \ldots,~ s_{(k+2)n},~ t_1,~ \ldots,~ t_{k+1}\rangle$
with $s_i\in T'_{k,n}$ and so
\[
  X~ (B^k B)^{(k+2)n} = \langle s_1,~ \ldots,~ s_{k+1},~ \langle s_{k+2},~ \ldots,~ s_{(k+2)n},~ t_1,~ \ldots,~ t_{k+1},~ \x\rangle\rangle.
\]
We can easily see $s_1,\ \ldots,\ s_{k+1}$, and $\langle s_{k+2},\ \ldots,\ s_{(k+2)n},\ t_1,\ \ldots,\ t_{k+1}, \ \x\rangle$ are in $T'_{k,n}$.
\end{proof}

From the definition of $T_{k,n}$, we can compute that $a(X)$ equals $k+1$ or $(k+2)n + k + 1$.
Particularly, we get the following:
\begin{lemma}
\lablem{targs}
For any $X \in T_{k,n}$, $a(X) \le (k+2)n + k + 1 = l((B^k B)^{(k+2)n}) - 1$.
\end{lemma}
This lemma is crucial to show that the number of variables in $((B^k B)^{(k+2)n})_{(i)}$ is monotonically non-decreasing.
Put $Z = (B^k B)^{(k+2)n}$ for short.
Since $Z \in T_{k,n}$,
we have $\{Z_{(i)} \mid i \ge 1\} \subset T_{k,n}$ by \reflem{closed}.
Using \reflem{targs}, 
we can simplify \reflem{asymp} in the case where $X = Z_{(i)}$ and $X'=Z$ as follows:
\begin{align}
  l(Z_{(i+1)}) &= l(Z_{(i)}) + (k+2)n + k + 1 - a(Z_{(i)}) \labeqn{lasymp}\\
  a(Z_{(i+1)}) &= a(N_1(Z_{(i)})) + k + 1 \labeqn{aasymp}\\
  N_1(Z_{(i+1)}) &=
  \begin{cases}
    N_2(Z_{(i)}) & (\text{if $N_1(Z_{(i)})$ is a variable})\\
    N_1(N_1(Z_{(i)})) & (\text{otherwise}).
  \end{cases} \labeqn{nasymp}
\end{align}
By \refeqn{lasymp} and \reflem{targs}, we get
$l(Z_{(i+1)}) \ge l(Z_{(i)})$.

To prove that $Z$ does not have the $\rho$-property,
it suffices to show the following:
\begin{lemma}
For any $i \ge 1$, there exists $j > i$ that satisfies 
$l(Z_{(j)}) > l(Z_{(i)})$.
\end{lemma}
\begin{proof}
Suppose that there exists $i \ge 1$ that satisfies $l(Z_{(i)}) = l(Z_{(j)})$ for any $j > i$.
We get $a(Z_{(j)}) = (k+2)n + k + 1$ by \refeqn{lasymp} and then $a(N_1(Z_{(j)})) = (k+2)n$ by \refeqn{aasymp}.
Therefore $N_1(Z_{(j)})$ is not a variable for any $j>i$ and
from \refeqn{nasymp}, we obtain $N_1(Z_{(j)}) = N_1(N_1(Z_{(j-1)})) = \cdots = \underbrace{N_1(\cdots N_1(}_{j-i+1}Z_{(i)})\cdots)$ for any $j > i$.
However, this implies that $Z_{(i)}$ has infinitely many variables and it contradicts.
\end{proof}

Now, we get the desired result:
\begin{theorem}
\labthr{b2anti}
For any $k \ge 0$ and $n > 0$, $(B^k B)^{(k+2)n}$ does not have the $\rho$-property.
\end{theorem}
The key fact which enables to show the anti-$\rho$-property of $(B^k B)^{(k+2)n}$ is
the existence of the set $T_{k,n} \supset \left\{((B^k B)^{(k+2)n})_{(i)} \ \middle|\ i \ge 1 \right\}$ which satisfies \reflem{targs}.
In the same way to above, we can show the anti-$\rho$-property of a $B$-term which has such a ``good'' set. That is,
\begin{theorem}
\labthr{antibgeneral}
Let $X$ be a $B$-term and $T$ be a set of $B$-terms.
If $\left\{ X_{(i)} \ \middle|\  i \ge 1\right\} \subset T$ and $l(X) \ge a(X') + 1$ for any $X' \in T$,
then $X$ does not have the $\rho$-property.
\end{theorem}

Here is an example of the $B$-terms which satisfies the condition in \refthr{antibgeneral} with some set $T$.
Consider $X = (B^2 B)^2 \circ (B B)^2 \circ B^2
= \<\x,\ \<\x,\ \<\x,\ \<\x,\ \x,\ \x>,\ \x>,\ \x>>$.
We inductively define $T'$ as follows:
\begin{enumerate}
\item $\x \in T'$
\item For any $t \in T'$, $\langle \x,\ t,\ \x \rangle \in T'$
\item For any $t_1, t_2 \in T'$, $\langle \x,\ t_1,\ \x,\ \langle \x,\ t_2,\ \x \rangle,\ \x \rangle \in T'$
\end{enumerate}
Then $T = \left\{ \<t_1,\ \<\x,\ t_2,\ \x>> ~\middle|~ t_1, t_2 \in T' \right\}$ satisfies the condition in \refthr{antibgeneral}.
It can be checked simply by case analysis. Thus
\begin{theorem}
$(B^2 B)^2\circ (B B)^2 \circ B^2$ does not have the $\rho$-property.
\end{theorem}

\refthr{antibgeneral} gives a possible technique to prove the monotonicity with respect to $l(X_{(i)})$, or, the anti-$\rho$-property of $X$, for some $B$-term $X$.
Moreover, we can consider another problem on $B$-terms: ``Give a necessary and sufficient condition to have the monotonicity for $B$-terms.''

\section{Concluding remark}
\labsec{concl}
We have investigated the $\rho$-properties of $B$-terms in particular forms
so far.
While the $B$-terms equivalent to $B^n B$ with $n\leq 6$ have the $\rho$-property,
the $B$-terms $(B^k B)^{(k+2)n}$ with $k\geq0$ and $n>0$
and $(B^2 B)^2\circ(B B)^2\circ B^2$ do not.
In this section,
remaining problems related to these results are introduced
and possible approaches to illustrate them are discussed.

\subsection{Remaining problems}
The $\rho$-property is defined any combinatory terms
(and closed $\lambda$-terms).
We investigates it only for $B$-terms 
as a simple but interesting instance
in the present paper.
%
From his observation on repetitive right applications for several $B$-terms,
Nakano~\cite{Nakano08trs} has conjectured as follows.
\begin{conjecture}
\labcnj{B-conj}
$B$-term $e$ has the $\rho$-property if and only if 
$e$ is a monomial, \ie, $e$ is equivalent to $B^n B$ with $n\geq0$.
\end{conjecture}
The if-part for $n\leq 6$ has been shown by computation
and the only-if-part for $(B^k B)^{(k+2)n}~\allowbreak(k\ge0,n>0)$
and $(B^2 B)^2\circ(B B)^2\circ B^2$
has been shown by \refthr{b2anti}.
This conjecture implies that the $\rho$-property of $B$-terms is decidable.
We surmise that the $\rho$-property of even $BCK$- and $BCI$-terms
is decidable.
The decidability for the $\rho$-property of $S$-terms and $L$-terms can
also be considered.
Waldmann's work on a ration representation of normalizable $S$-terms may
be helpful to solve it.
We expect that none of $S$-terms have the $\rho$-property
as $S$ itself does not, though.
Regarding to $L$-terms,
%
Statman's work~\cite{Statman89sc} may be helpful
where
equivalence of $L$-terms is shown decidable
up to a congruence relation induced by $L~e_1~e_2\to e_1~(e_2~e_2)$.
It would be interesting to investigate
the $\rho$-property of $L$-terms in this setting.


\subsection{Possible approaches}
The present paper introduces
a canonical representation
to make equivalence check of $B$-terms easier.
The idea of the representation is based on that
we can lift all $\circ$'s (2-argument $B$)
to outside of $B$ (1-argument $B$)
by equation~\refeqn{Bo-distr}.
One may consider it the other way around.
Using the equation, 
we can lift all $B$'s (1-argument $B$)
to outside of $\circ$ (2-argument $B$).
Then one of the arguments of $\circ$ becomes $B$.
By equation~\refeqn{Bo-push},
we can move all $B$'s to right.
Thereby we find another canonical representation for $B$-terms
given by
\begin{align*}
e &::= B \mid B~ e \mid e \circ B
\end{align*}
whose uniqueness could be easily proved in a way similar to \refthr{canonical}.

Waldmann~\cite{Waldmann13email} suggests
that the $\rho$-property of $B^{n} B$ may be checked
even without converting $B$-terms into canonical forms.
He simply defines $B$-terms by
\begin{align*}
e &::= B^k \mid e~ e
\end{align*}
and regards $B^k$ as a constant
which has a rewrite rule
$B^k~e_1~e_2~\dots~e_{k+2}\to e_1~(e_2~\dots~e_{k+2})$.
He implemented a check program in Haskell
to confirm the $\rho$-property.
Even in the restriction on rewriting rules,
he found that
$\sapp{(B^0 B)}{9}=\sapp{(B^0 B)}{13}$,
$\sapp{(B^1 B)}{36}=\sapp{(B^1 B)}{56}$,
$\sapp{(B^2 B)}{274}=\sapp{(B^2 B)}{310}$ and
$\sapp{(B^3 B)}{4267}=\sapp{(B^3 B)}{10063}$,
in which it requires a bit more right applications
to find the $\rho$-property than the case of a canonical representation.
If the $\rho$-property of $B^n B$ for any $n\geq0$ is shown
under the restricted equivalence given by rewriting rules,
then we can conclude the if-part of \refcnj{B-conj}.


Another possible approach is to observe the change of (principal) types
by right repetitive application.
Although there are many distinct $\lambda$-terms of the same type,
we can consider a desirable subset of typed $\lambda$ terms.
As shown by Hirokawa~\cite{Hirokawa93tcs},
each $BCK$-term can be characterized by its type,
that is, 
any two $\lambda$-terms in $\CL{BCK}$ of the same principal type are identical
up to $\beta$-equivalence.
This approach may require to observe unification between types
in a clever way.








\end{document}